\documentclass{fundam}

\publyear{24}
\papernumber{2102}
\volume{193}
\issue{1}

    \versionForARXIV

\usepackage{url} 
\usepackage[ruled,lined]{algorithm2e}
\usepackage{graphicx}
\usepackage{tikz}
\usetikzlibrary{automata, positioning, arrows}

\begin{document}

\title{Maximum Centre-Disjoint Mergeable Disks}

\author{Ali Gholami Rudi\thanks{Address for correspondence:  Department of Electrical and Computer Engineering,
                        Babol Noshirvani University of Technology, Babol, Mazandaran, Iran. \newline \newline
                    \vspace*{-6mm}{\scriptsize{Received January 2023; \ accepted July 2024.}}}
\\
Department of Electrical and Computer Engineering \\
Babol Noshirvani University of Technology \\
Babol, Mazandaran, Iran\\
gholamirudi@nit.ac.ir
}

\maketitle

\runninghead{A.G. Rudi}{Maximum Centre-Disjoint Mergeable Disks}

\begin{abstract}
Given a set of disks in the plane, the goal of the problem
studied in this paper is to choose a subset of these disks
such that none of its members contains the centre of any other.
Each disk not in this subset must be merged with one of
its nearby disks that is, increasing the latter's radius.
This problem has applications in labelling rotating maps and
in visualizing the distribution of entities in static maps.
We prove that this problem is NP-hard.
We also present an ILP formulation for this problem,
and a polynomial-time algorithm for the special case
in which the centres of all disks are on a line.
\end{abstract}


\section{Introduction}
A motivating example for the problem studied in this paper is
the following about drawing text labels on a digital map that
can be rotated: suppose there are a number of points on the map
that represent map features.  To each of these \emph{feature points}
a text label is assigned that describes the feature,
 like the name of a junction. When the map is rotated by the user, these labels must remain
horizontal \linebreak  for the sake of readability, and therefore, they
are rotated in the$\,$ reverse direction around their feature
 \eject \noindent
   point (see the first two parts of Figure~\ref{rota}).
Labels are difficult to read if they overlap, and therefore,
only a non-overlapping subset of the labels are drawn on the map.
If a label cannot be drawn because it overlaps with other
labels, the text of its label must be appended to a nearby label that
is drawn.
The goal is to draw the maximum number of labels on the map such
that none of them overlap when rotating the map.
This is demonstrated in Figure~\ref{rota}.
Part~(a) shows four feature points and their labels.
Part~(b) shows the map when it is rotated 45 degrees counterclockwise;
instead of rotating the map, the labels are equivalently rotated 45
degrees clockwise.  Obviously the two labels on the left side of the
map overlap, making the map of Part~(a) infeasible.  Part~(c) shows what
happens when these labels are merged.
The remaining three labels never overlap when the map is rotated.

\begin{figure}[h]
\vspace*{2mm}
	\centering
	\includegraphics[width=\columnwidth]{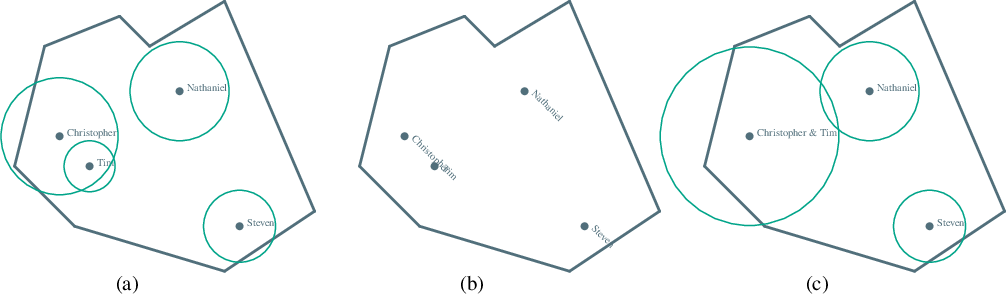}
	\caption{An example rotating map with 4 labels.
	(a) The initial configuration in which two of the labels overlap during rotation (the circles show the area covered by the labels during rotation).
	(b) After rotating the map 45 degrees counterclockwise.
	(c) Two of the labels are merged so that none of the label overlap during rotation.}
	\label{rota}\vspace*{-2mm}
\end{figure}

Placing as many labels as possible on a map (known as map labelling) is
a classical optimization problem in cartography and graph drawing \cite{formann91}.
For static maps, i.e.\ maps whose contents does not change, the problem of placing labels on a map can be stated
as an instance of geometric independent set problem (also known as
packing for fixed geometric objects): given a set of geometric
objects, the goal is to find its largest non-intersecting subset.
In the weighted version, each object also has a weight and the
goal is to find a non-intersecting subset of the maximum possible weight.

A geometric intersection graph, with a vertex for each object and
an edge between intersecting objects, converts this geometric problem to
the classical maximum independent set for graphs, which is
NP-hard and difficult to approximate even within a factor of $n^{1-\epsilon}$,
where $n$ is the number of vertices and $\epsilon$ is any non-zero positive constant \cite{hastad96}.
Although the geometric version remains NP-hard even for unit disks \cite{fowler81},
it is easier to approximate, and several polynomial-time approximation
schemes (PTAS) have been presented for this problem \cite{hochbaum85,agarwal98,erlebach05,chan03,chan12}.

Dynamic maps allow zooming, panning, or rotation, and labelling
in such maps seems more challenging.
Most work on labelling dynamic maps consider zooming and panning
operations \cite{been06}.
Gemsa et al.~\cite{gemsa11} were the first to formally study labelling
rotating maps.  With the goal of maximising the total duration in
which labels are visible without intersecting other labels,
they proved the problem to be NP-hard and presented a $1/4$-approximation
algorithm and a PTAS, with the presence of \mbox{restrictions} on the distribution
of labels on the map.  Heuristic algorithms and Integer Linear Programming
(ILP) formulations have also been presented for this problem~\cite{gemsa16,cano17}.
Note that in these problems, invisible labels do not get merged with
visible labels.
Yokosuka and Imai~\cite{yokosuka13} examined a variant of this problem,
in which all of the labels are always present in the solution
and the goal is to maximise their size.

A related problem is crushing disks \cite{funke16},
in which a set of prioritized disks are given as input,
whose radii grow over time, as map labels do when zooming in.  When
two disks touch, the one with the lower priority disappears.  The
radii of the disks grow linearly, and when a disk disappears, the
radius of the other disk does not change.  The goal is to find the
order in which disks disappear and the process finishes when only one
disk remains.

In this paper, we investigate a problem similar to geometric
independent set for a set of disks, except that
i) the disks in the output must be centre-disjoint (none
of them can contain the centre of another) but they
may overlap,
ii) each disk that does not appear in the output must be
merged with a disk, containing its centre, that does.
When a disk is merged with another, the radius of the latter
is increased by the radius of the former.
Also to preserve the locality of the merges,
a disk $A$ can be merged with another disk $B$, only if
all disks closer to $B$ than $A$ (considering the distance
between disk centres) are also merged with $B$,
and after merging these closer disks, $B$ must contain
the centre of $A$.
This problem is formally defined in Section~\ref{spre}.
We prove this problem to be NP-hard via a reduction from
Planar Monotone 3-SAT \cite{berg10}.
Note that even without this restriction on merge order,
the problem remains NP-hard; we have presented a PTAS in an
earlier paper \cite{rudi21} for the case in which this
restriction on merge order is not assumed.

To observe how the introductory example at the beginning of this section
reduces to this problem, consider the disks in Figure~\ref{rota} (a).
The disk centred at each feature point shows the region covered by its
label during rotation.  Only if a disk contains the centre of another,
their corresponding labels intersect at some point during rotation.
As another application of this problem, centre-disjoint disks can
show the distribution of facilities in an area.  For instance,
Figure~\ref{rotx} shows the distribution of schools in Munich.
It was obtained by placing a disk of radius 50 meters
on each school (the coordinates of schools were obtained from
OpenStreetMap data).  Then, an ILP (Section~\ref{silp}) was
used to obtain the maximum number of centre-disjoint disks in our problem

\begin{figure}
	\centering
	\includegraphics[height=8cm]{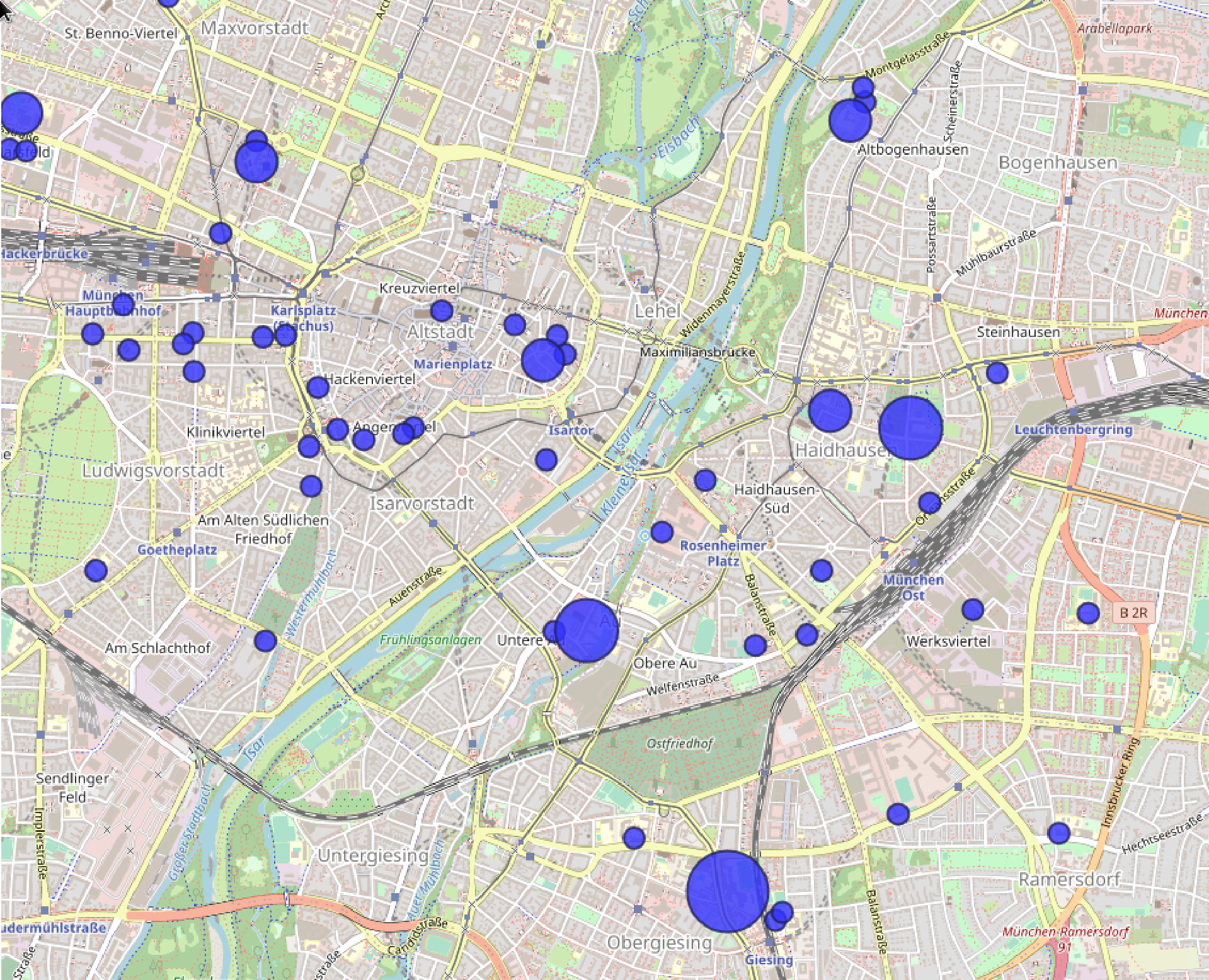}
	\caption{The distribution of schools in Munich; disks
	corresponding to neighbouring schools were merged
	using the ILP of Section~\ref{silp}
	to obtain larger, centre-disjoint disks.}
	\label{rotx}
\end{figure}

Note that the centre-disjointness property of disks, which we assume
for the output of our problem, is also used in the definition of
transmission graphs of a set of disks, in which a vertex is assigned
to each disk and a directed edge from a disk to another shows that the
former contains the centre of the latter \cite{kaplan19}.
These graphs have been studied, for instance, for computing their spanners
(a subgraph to estimate the distance between pairs of vertices) \cite{kaplan18},
counting their number of triangles and computing their girth \cite{kaplan19},
and their recognition \cite{klost18}.
Transmission graphs show the static relation between input disks
and do not directly help us in our problem,
in which the goal is to find a set of radius-changing disk merges,
which makes them centre-disjoint.

To obtain an efficient algorithm for the problem, we also examine a
more restricted version, in which the centres of input disks
are on a line (Section~\ref{salg}).  For this version, we present a
polynomial-time algorithm that incrementally obtains a set of
centre-disjoint disks with the maximum size.
Note that the assumption of collinear inputs have been applied
to many other challenging problems, such as~\cite{biniaz18}.

This paper is organised as follows.
In Section~\ref{spre} we introduce the notation used in this paper
and formally state the problem.  Then, in Section~\ref{shard} we
show that the problem studied in this paper is NP-hard.
In Section~\ref{salg}, we present algorithms for solving this
problem: we present an ILP formulation for solving the general
case of the problem, and a polynomial-time
dynamic programming algorithm for the case in which all disk
centres are on a line.
Finally, in Section~\ref{scon} we conclude this paper.

\section{Notation and preliminary results}
\label{spre}
Let $D = \{ d_1, d_2, \dots , d_n \}$ be a set of $n$ disks.
The radius of $d_i$ is denoted as $r_i$,
and its centre is denoted as $p_i$.

\begin{definition}
\label{dassign}
A function $\phi$ from $D$ to itself is an \emph{assignment},
if $\phi(\phi(d_i))$ is $\phi(d_i)$ for every $d_i$ in $D$.
According to an assignment $\phi$, the disks in $D$
can be either \emph{selected} or \emph{merged}:
if $\phi(d_i)$ is $d_i$, the disk $d_i$ is selected, and otherwise,
it is merged.
The cardinality of an assignment, denoted as $\left|\phi\right|$,
is the number of selected disks in $\phi$.
\end{definition}

The relation defined by assignments (Definition~\ref{dassign})
describes disk merges in our problem.
For any disk $d_i$, if we have $\phi(d_i) = d_j$ and $i \neq j$,
it implies that $d_i$ is merged with $d_j$.
On the other hand, the relation $\phi(d_i) = d_i$ implies
that $d_i$ is a selected disk and is not merged with any other disk.
Since a disk can be merged with selected disks only,
for any disk $d_i$, we have $\phi(\phi(d_i)) = \phi(d_i)$.

\begin{definition}
\label{ddelta}
The aggregate radius of a selected disk $d_i$ with
respect to an assignment $\phi$, denoted with some misuse of notation
as $r_i(\phi)$, is the sum of its radius and that of every disk
merged with it, or equivalently,
$$
r_i(\phi) = \sum_{j:~\phi(d_j) = d_i} r_j.
$$
Let $\delta_i$ be the sequence of disks in $D \setminus \{d_i\}$,
ordered increasingly by the distance of their centres from the
centre of $d_i$, and let $\delta_i(j)$ denote its $j$-th disk.
The $j$-th aggregate radius of $d_i$, denoted as $r_i(j)$, is defined as
its aggregate radius if $\{\delta_i(1), \delta_i(2), \dots , \delta_i(j)\}$ are
merged with $d_i$.
\end{definition}

We now define proper assignments (Definition~\ref{dproper}).
In the rest of this paper, the distance between two disks
is defined as the Euclidean distance between their centres.

\begin{definition}
\label{dproper}
An assignment $\phi$ is \emph{proper} if it meets the following conditions.
\begin{enumerate}
\itemsep=0.9pt
\item
The disk $\delta_i(j)$ can be merged with $d_i$, only if
$\delta_i(k)$, for every $k$ where $1 \le k < j$, are also merged with $d_i$.
In other words, all disks closer to $d_i$ than $\delta_i(j)$
are also merged with $d_i$.
\item
The disk $\delta_i(j)$ can be merged with $d_i$, only if
the distance between the centre of $d_i$ and $\delta_i(j)$
is less than $r_i(j-1)$.  In other words,
after merging $\delta_i(k)$ for $1 \le k < j$, $d_i$ must contain the
centre of $\delta_i(j)$.
\item
Selected disks must be centre-disjoint with respect to their aggregate radii;
i.e.~none of them can contain the centre of any other selected disk.
More precisely, for indices $i$ and $j$ such that $i \neq j$, $\phi(d_i) = d_i$,
and $\phi(d_j) = d_j$, we have $\left| p_i p_j \right| \geq \max (r_i (\phi), r_j(\phi))$.
\end{enumerate}
\end{definition}

Note the first two items in Definition~\ref{dproper} ensure the locality
of the merges, which is especially important in the labelling application
mentioned in the Introduction.

\begin{definition}
\label{dmcmd}
Given a set of disks, in the Maximum Centre-Disjoint Mergeable Disks Problem (MCMD),
the goal is to find a proper assignment of the maximum possible cardinality.
\end{definition}

\begin{figure}[!h]
	\centering
	\includegraphics[width=5.5cm]{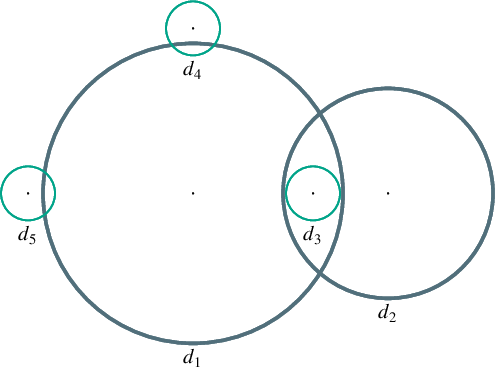}
	\caption{An example set of disks with two proper assignments:
	Either all disks are merged with $d_1$, which gives a proper assignment
	of cardinality 1, or just $d_3$ is merged with $d_2$, which gives a
	proper assignment of cardinality 4.  The latter is a solution to
	MCMD, because it has the maximum cardinality.
	}
	\label{rot1}
\end{figure}

Figure~\ref{rot1} shows a configuration of five disks with more than
one proper assignment.  Disk $d_3$ can be merged with $d_1$, after which,
$d_1$ would contain the centre of $d_4$ and $d_5$, both of which
then have to be merged with $d_1$.  These merges result in $d_1$ containing
the centre of $d_2$, which would also be merged.  Therefore,
in this assignment $\phi_1$, we have $\phi_1(d_i) = d_1$,
for $1 \le i \le 5$, and its cardinality is one.
Alternatively, in assignment $\phi_2$ we can merge $d_3$ with
$d_2$, as the latter contains the centre of the former.
The remaining disks are centre-disjoint.  Therefore,
we have $\phi_2(d_1) = d_1$, $\phi_2(d_2) = d_2$,
$\phi_2(d_3) = d_2$, $\phi_2(d_4) = d_4$, $\phi_2(d_5) = d_5$,
and its cardinality is four.
Assignment $\phi_2$ maximises the number of selected disks,
and is a solution to MCMD for the configuration of disks in Figure~\ref{rot1}.

\medskip
Not every set of disks has a proper assignment.
Figure~\ref{rot2} shows an example.
Disk $d_3$ can be merged with either $d_1$ or $d_2$.
If $d_3$ is merged with $d_1$, $d_5$ cannot be merged with $d_2$,
because of the second condition of proper assignments: $d_5$
can be merged with $d_2$, only if all closer disks to $d_2$ are
merged with it (but $d_3$ which is closer to $d_2$ than $d_5$
is not).  Therefore, $d_5$ can be neither merged, nor selected (because
its centre is contained in $d_2$).
Similarly, if $d_3$ is merged with $d_2$, $d_4$ can neither be
merged nor selected.  Thus, there exists no proper assignment
for these set of disks.  In Section~\ref{srelax} we introduce a
variant of MCMD by relaxing the second condition of Definition~\ref{dproper},
in which every instance has a solution.

\begin{figure}[h]
\vspace*{2mm}
	\centering
	\includegraphics[width=6cm]{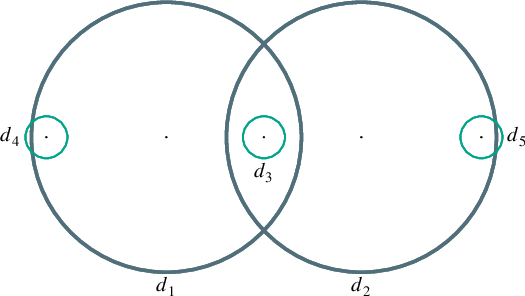}
	\caption{An example set of disks with no proper assignment.
	Either $d_3$ and $d_4$ are merged with $d_1$, after which
	$d_5$ cannot be merged with $d_2$ (but must be),
	or $d_3$ and $d_5$ are merged with $d_2$, after which $d_4$ cannot
	be merged with $d_1$ (but must be).
	}\vspace*{-2mm}
	\label{rot2}
\end{figure}

\section{Hardness of maximum centre-disjoint mergeable disks}
\label{shard}
Instead of proving that the decision version of MCMD (Definition~\ref{ddmcmd})
is NP-complete, we show that even deciding whether a set of disks has a proper
assignment (Definition~\ref{dpropmcmd}) is NP-complete (clearly
the latter implies the former).
To do so, we perform a reduction from the NP-complete
\textsc{Planar Monotone 3-SAT} (Definition~\ref{d3sat}) \cite{lichtenstein82}
to \textsc{Proper MCMD} (Definition~\ref{dpropmcmd}).

\subsection{Hardness of MCMD}
\begin{definition}
\label{ddmcmd}
In the $k$-\textsc{MCMD} problem, we are given a set of disks and we
have to decide if there exists a proper assignment of cardinality
at least $k$ or not.
\end{definition}

\begin{definition}
\label{dpropmcmd}
In the \textsc{Proper MCMD} problem, we are given a set of disks and we
have to decide if there exists a proper assignment.
\end{definition}

\begin{definition}
\label{d3sat}
\textsc{Monotone 3-SAT} is a variant of 3-SAT, in which all variables
of each clause are either positive or negative.
An instance of \textsc{Monotone 3-SAT} is called \textsc{Planar},
if it can be modeled as a planar bipartite graph with parts $V$ corresponding
to variables and $C$ corresponding to clauses; each vertex in $C$ is incident
to at most three variables, which correspond to the variables that appear in
the clause.
\end{definition}

\begin{figure}[!h]
\vspace*{2mm}
	\centering
	\includegraphics[width=13cm]{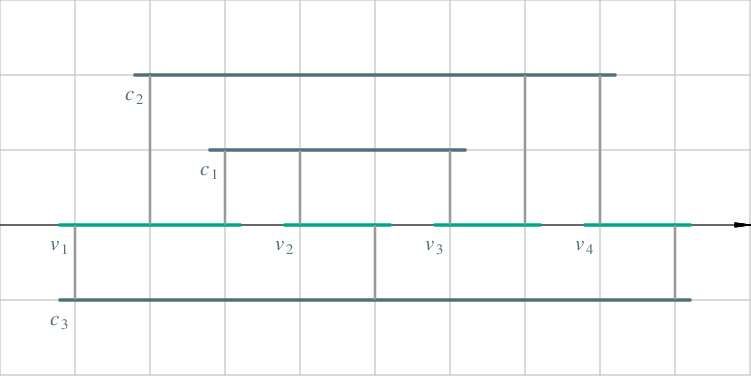}
	\caption{A monotone rectilinear representation of a \textsc{Planar Monotone 3-SAT} instance with three clauses and four variables.
	Horizontal segments on the $x$-axis denote the variables, and horizontal segments above
	and below the $x$-axis denote positive and negative clauses, respectively:
	$c_1 = v_1 \lor v_2 \lor v_3$, $c_2 = v_1 \lor v_3 \lor v_4$,
	$c_3 = \neg v_1 \lor \neg v_2 \lor \neg v_4$.
	}
	\label{rotd}
\end{figure}

Deciding if an instance of \textsc{Planar Monotone 3-SAT} is
satisfiable is NP-complete~\cite{berg10}.
It can be proved that every instance of \textsc{Planar Monotone 3-SAT} has
a \emph{monotone rectilinear representation} (Definition~\ref{drect};
an example is shown in Figure~\ref{rotd}), and
also, if for every instance of \textsc{Planar Monotone 3-SAT} its
monotone rectilinear representation is also given, the problem
remains NP-Complete~\cite{berg10}.

\begin{definition}
\label{drect}
A \emph{monotone rectilinear representation} of an instance of
\textsc{Planar Monotone 3-SAT} is a drawing of the instance
with the following properties:
i) Variable are drawn as disjoint horizontal segments on the $x$-axis,
ii) positive clauses are drawn as horizontal segments above the $x$-axis,
iii) negative clauses are drawn as horizontal segments below the $x$-axis,
iv) an edge is drawn as a vertical segment between a clause segment and
the segments corresponding to its variable, and
v) the drawing is crossing-free.
\end{definition}

Figure~\ref{rotd} shows a monotone rectilinear representation of an instance of
\textsc{Planar Monotone 3-SAT} with three clauses.  Lemma~\ref{lgrid} shows
how to map a \textsc{Planar Monotone 3-SAT} to a two-dimensional integer grid
(its proof is presented at the end of this section).

\begin{lemma}
\label{lgrid}
For an instance of \textsc{Planar Monotone 3-SAT} with $v$ variables
and $c$ clauses, there exists a monotone rectilinear representation on
a two-dimensional integer grid with $c + 1$ rows and $3c + v$ columns,
such that horizontal segments, which
represent variables and clauses, appear on horizontal grid lines, and
vertical segments appear on vertical grid lines.
\end{lemma}

To prove that \textsc{Proper MCMD} is NP-complete, we perform a reduction from
\textsc{Planar Monotone 3-SAT} to \textsc{Proper MCMD}.
To do so, we create an instance of \textsc{Proper MCMD} from the monotone
rectilinear representation of any instance of \textsc{Planar Monotone 3-SAT}
in Theorem~\ref{tnp}.
In our construction, we use two types of disks:
\begin{itemize}
\item Disks, which by our construction, are always
selected (their centres can never be inside any other disk).
We call them \emph{s-disks} for brevity.
\item
Disks of very small radius, which are contained in at least one s-disk, and thus,
are surely merged in our construction.  We call these disks \emph{m-disks}.
We assume that the radius of m-disks is so small compared to the radius
of s-disks that after merging any number of m-disks with an s-disk,
the centre of no new disk would enter the s-disk in our configuration.
In the instance of \textsc{Proper MCMD} that we construct,
each s-disk contains at least one m-disk.
\end{itemize}

We create a configuration of disks using \emph{gadgets}, each of
which consists of some m-disks and s-disks.
The m-disks of a gadget are either internal (internal m-disks)
or can be shared with other gadgets (shared m-disks).
Parts (a) and (b) of Figure~\ref{rotc} show two gadgets (from each gadget,
only an s-disk and an m-disk is shown).
In Figure~\ref{rotc} (c) these two gadgets are joined at m-disk $m$.
In a proper assignment, $m$ is merged either with
an s-disk of $A$ or with an s-disk of $B$.
With respect to gadget $A$, if $m$ is merged with $A$ in a proper
assignment, we say that it is \emph{merged in}, and otherwise,
\emph{merged out} with respect to $A$.

\begin{figure}[!ht]
	\centering
	\includegraphics[width=8cm]{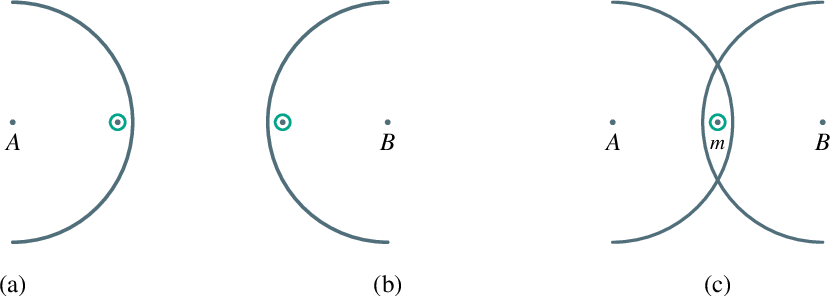}
	\caption{Two gadgets (Parts (a) and (b)), joined at one of their m-disks (Part (c)).}
	\label{rotc}
\end{figure}

We use the following gadgets in our construction.
The gadgets and the distance between shared m-disks
of each of them are shown in Figure~\ref{fgadgets};
s-disks (denoted as $s_i$) are large disks and
m-disks (denoted as $m_i$) are small disks.
More details about these gadgets are provided now:

\begin{figure}[ht]
 \vspace*{2mm}
	\includegraphics[width=0.97\columnwidth]{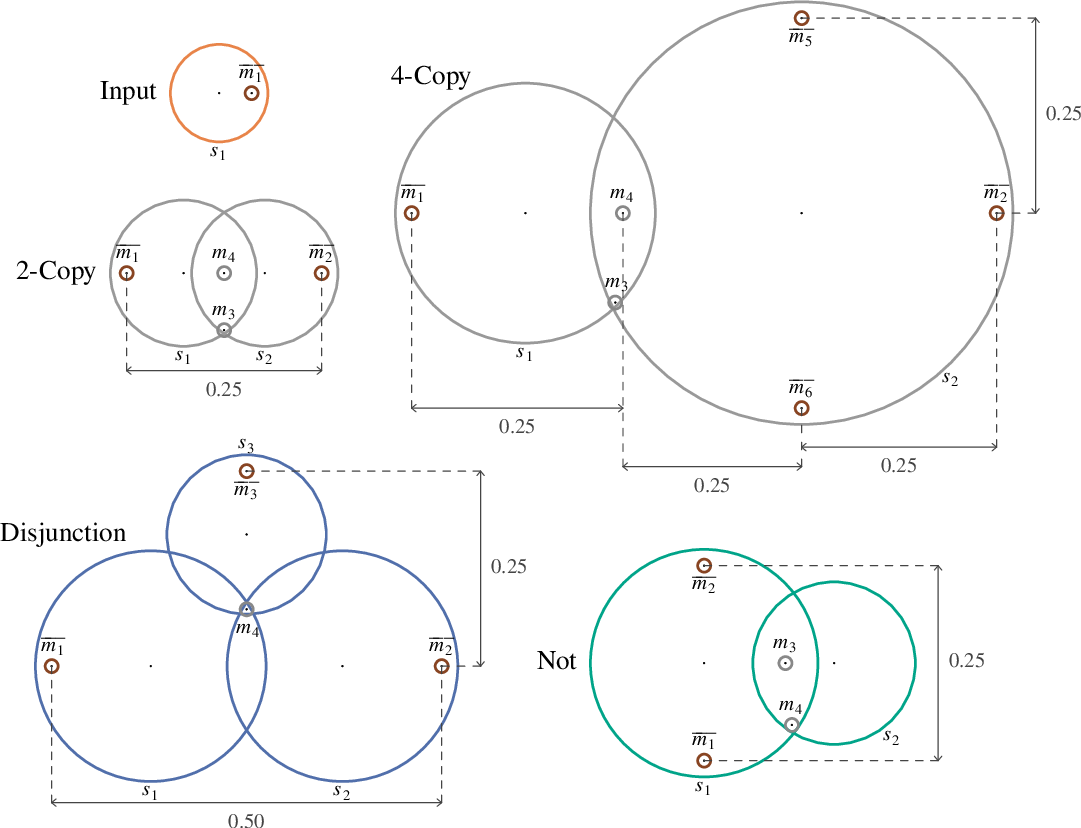}
	\centering
	\caption{Gadgets used in the proof of Theorem~\ref{tnp};
	$s_i$ and $m_i$ for different values of index $i$
	denote s-disks and m-disks, respectively.
	Shared m-disks are indicated with overlines (like $m_1$ in Input).
	The sizes of the gadgets are specified such that the gadgets
	fit together in the proof of Theorem~\ref{tnp}.}
	\label{fgadgets}
\end{figure}

\begin{itemize}
\item Input: This gadget has only one shared m-disk, which can be either merged in or merged out.
\item Copy: We use two gadgets for copy in our construction: one with
two m-disks and one with four (both of them are demonstrated in Figure~\ref{fgadgets}),
which we reference as 2-Copy and 4-Copy, respectively.
The logic behind both of them is similar and is explained thus.
If $m_1$ is merged in, $m_2$ (also $m_5$ and $m_6$ if present) is merged out,
and if $m_1$ is merged out, $m_2$ (also $m_5$ and $m_6$) is merged in.
To see why, note that $m_3$ can be merged either with $s_1$ or with $s_2$.
If $m_3$ is merged with $s_1$, both $m_1$ and $m_4$ must also
be merged with $s_1$, because $m_3$ is farther than both to $s_1$.
Since $m_4$ is merged with $s_1$, $m_2$ (also $m_5$ and $m_6$)
cannot be merged with $s_2$ and therefore they must be merged out.
Similarly, if $m_3$ is merged with $s_2$, $m_2$ (also $m_5$ and $m_6$)
must be merged with $s_2$ as well, and $m_1$ must be merged out.
\item Disjunction: One or more of its shared m-disks are merged in.
Clearly, $m_4$ must be merged with $s_1$, $s_2$, or $s_3$.
If it is merged with $s_i$ ($i \in \{1, 2, 3\}$), $m_i$
must also be merged with $s_i$, and $m_j$ ($j \ne i$) may or
may not be merged in.
\item Not: Either both $m_1$ and $m_2$ are merged in or both of them
are merged out.
This is because $m_4$ can be merged either with $s_1$ or $s_2$.
If it is merged with $s_1$, m-disks $m_1$, $m_2$, and $m_3$ must
also be merged with $s_1$, because $m_4$ is farther than all of them.
Otherwise, if $m_4$ is merged with $s_2$, m-disk $m_3$ must also
be merged with $s_2$ and therefore, none of $m_1$ and $m_2$ can be
merged with $s_1$, because $m_3$ (which is closer than both)
is not merged with $s_1$.  Thus, $m_1$ and $m_2$ must merge out.
\end{itemize}
In our construction of the proof of Theorem~\ref{tnp}, some
of the shared m-disks of 4-Copy gadget are unused and are
not shared with any other gadget; for such instances of 4-Copy,
their unused shared m-disks must be removed.

\begin{theorem}
\label{tnp}
\textsc{Proper MCMD} is NP-complete.
\end{theorem}
\begin{proof}
It is trivial to show that \textsc{Proper MCMD} is in NP.
To show that it is NP-hard, we reduce \textsc{Planar Monotone 3-SAT}
to \textsc{Proper MCMD}.
Let $I$ be an instance of \textsc{Planar Monotone 3-SAT},
with variables $V$ and clauses $C$.
Based on Lemma~\ref{lgrid}, there exists a monotone rectilinear
representation of $I$ on a
$(\left|C\right| + 1) \times (3 \cdot \left|C\right| + \left|V\right|)$
integer grid.  Let $R$ denote this representation.

\begin{figure}[!b]
\vspace*{4mm}
	\includegraphics[width=\columnwidth]{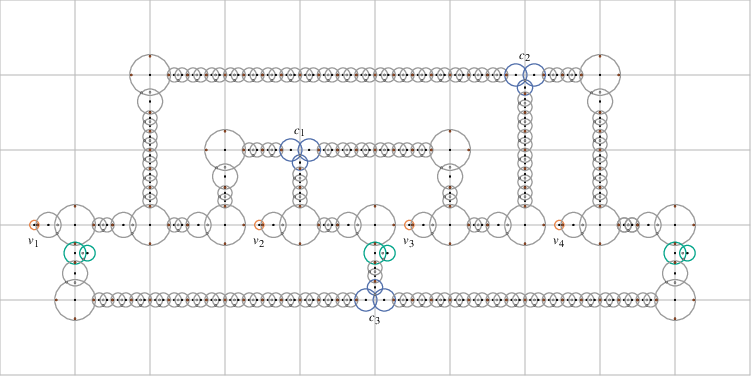}
	\centering
	\caption{A \textsc{Proper MCMD} instance obtained from the \textsc{Planar Monotone 3-SAT} instance of Figure~\ref{rotd}.}
	\label{rotf}\vspace{1mm}
\end{figure}

\medskip
We create an instance of \textsc{Proper MCMD} from $R$ as follows.
The transformation is demonstrated in Figure~\ref{rotf}, which
corresponds to the monotone rectilinear representation of Figure~\ref{rotd}.
\begin{enumerate}
\item We replace the segment corresponding to a variable in $R$
with an Input gadget and a chain of 2-Copy gadgets.  For each
intersection of this segment with a vertical segment, a 4-Copy
gadget is used.
\item Let $s$ be a horizontal segment corresponding to a clause in $R$.
Three variables appear in the clause, for each of which there is a
vertical segment that connects $s$ to a variable segment.
For the first and last intersections, 4-Copy gadgets are
used.  For the 2nd intersection, we use a Disjunction gadget.
These gadgets are connected using two chains of 2-Copy gadgets.
\item For each vertical segment that connects a variable segment
to a clause segment above the $x$-axis, we use a chain of 2-Copy
gadgets to connect the 4-Copy gadget of the variable segment to
the 4-Copy or Disjunction gadget (if it is the 2nd intersection)
of the clause segment.
For segments that appear below the $x$-axis, we do likewise,
except that we place a Not gadget before the chain of 2-Copy gadgets.
\end{enumerate}
Note that some of the gadgets of Figure~\ref{fgadgets} need to be
rotated or mirrored, for instance, in the vertical chains that connect
clauses and variables.  Also note that based on the sizes
shown in Figure~\ref{fgadgets}, shared m-disks always appear on
grid lines in our construction.  Since the distance between the
shared m-disks of any of the gadgets is at least 0.25,
at most four gadgets can appear on a grid segment of unit length.
Given that the total area of the grid, and therefore the total length of grid segments,
is bounded by $O(\left|C\right|^2)$,
the number of gadgets used in the resulting instance of \textsc{Proper MCMD}
is at most $O(\left|C\right|^2)$.
Thus, the size of the resulting \textsc{Proper MCMD} instance is
polynomial in terms of the size of the input \textsc{Planar Monotone 3-SAT}
instance.

Suppose there is a proper assignment for our \textsc{Proper MCMD} instance.
We obtain an assignment $A$ of the variables of our \textsc{Planar Monotone 3-SAT}
instance as follows.  We assign one to a variable if the m-disk of its
corresponding Input gadget is merged out, and assign zero otherwise.
Consider any clause $c$ in our \textsc{Planar Monotone 3-SAT} instance.
Let $g$ be the Disjunction gadget corresponding to $c$.

If $c$ is a positive clause, a chain of 2-Copy gadgets connects
the Input gadget of each of the variables that appear in $c$ to $g$.
Therefore, if variable $v$ appears in the
clause and if the shared m-disk of the Input gadget corresponding
to $v$ is merged out, the m-disk of the last 2-Copy gadget of its
chain is merged in inside $g$.  Since, one or more of the
shared m-disks of $g$ are merged in, at least one of the
literals in $g$ is satisfied.
Similarly, if $c$ is a negative clause, because there is a
Not gadget in the chain that connects each variable $v$ of $c$
to its Disjunction gadget, if the shared m-disk of the Input gadget corresponding
to $v$ is merged out, the m-disk of the last 2-Copy gadget of its
chain is also merged out inside $g$.
Since, one or more of the shared m-disks of $g$ are merged in,
at least one of the variables in $g$ is not satisfied.

Therefore, the \textsc{Planar Monotone 3-SAT} instance is satisfied
with assignment $A$.

\medskip
For the reverse direction, suppose there exists an assignment $A$ of
the variables, for which all clauses of $I$ are satisfied.
We can obtain a proper assignment in our \textsc{Proper MCMD} instance as follows.
For each variable $v$ in $V$, if $v$ is one, the shared m-disk of
the Input gadget corresponding to $v$ is merged out, and otherwise,
it is merged in.
Let $c$ be a positive clause in which variable $v$ with value one
appears (since $c$ is satisfied in $A$, variable $v$ must exist),
and let $g$ be the disjunction gadget corresponding to clause $c$.
Since $v$ is merged out, the m-disk of the last 2-Copy gadget
that connects the gadget corresponding to $v$ to $g$ is merged in with
respect to $g$.
This implies that one of the shared m-disks of the Disjunction gadget of
each positive clause is merged in.
We can similarly show that at least one of the shared m-disks of the
Disjunction gadgets corresponding to negative caluses are also merged in.
This yields a proper assignment for the \textsc{Proper MCMD} instance.
\end{proof}

In Corollary~\ref{csame} we show that even if all disks have the same
radius, the problem remains NP-hard.

\begin{corollary}
\label{csame}
\textsc{Proper MCMD} remains NP-complete, even if all disks are
required to be of the same radius.
\end{corollary}
\begin{proof}
We fix the radius of m-disks to $r = 0.01$.
We use the same construction as Theorem~\ref{tnp},
with the difference that we replace each s-disk with a number of
smaller disks of radius $r$ with the same centre, so that the sum of
the radii of these smaller disks equals the radius of the s-disk.
Since the disks added for each s-disk are not centre-disjoint,
and their centre cannot be contained in some other disk, exactly one
of them must be selected and after merging others, it reaches
the size of the original s-disk.  The rest of the proof of
Theorem~\ref{tnp} applies without significant changes.
\end{proof}

In the proof of Corollary~\ref{csame},
we can adjust the position of the disks that replace each s-disk
so that their centres do not coincide: they can be placed evenly
on a very short line segment (for instance of length 0.0001).
However, that they cannot be centre-disjoint, as they are to be merged.

Now we present the proof of Lemma~\ref{lgrid}.
\begin{proof}
Let $R$ be a monotone rectilinear representation of a \textsc{Planar
Monotone 3-SAT} instance (such a representation certainly exists~\cite{berg10}).
By extending horizontal segments of
$R$ we get at most $c + 1$ lines: one for the variables (the $x$-axis)
and at most $c$ for clauses.  Let $\ell_1, \ell_2, \dots, \ell_m$ be the lines
that appear above the $x$-axis ordered by their $y$-coordinates.
We move them (together with the segments appearing on them) so that,
$\ell_i$ is moved to $y = i$; vertical segments that connect them to a segment on
the $x$-axis may need to be shortend or lengthened during the movement.
Given that the $x$-coordinate of the end points of horizontal segments,
and also the vertical order of the segments, do not change, no new
intersection is introduced by this transformation.
The same is done for the lines that appear below the $x$-axis.

Repeating the same process for vertical segments,
we get at most $3c$ vertical lines.
We can similarly move these lines and the segments on them horizontally
so that they appear in order and consecutively on vertical integer grid lines.
Variables that do not appear in any clause, can be placed in at most
$v$ additional vertical grid lines.
This results in a $(c + 1) \times (3c + v)$ grid.
\end{proof}

\subsection{Relaxing merge order}
\label{srelax}
Due to the first condition of proper assignments
(Definition~\ref{dproper}), in a proper assignment $\phi$ of a set of
disks $D$, a disk $d_i$ can be merged with another disk $d_j$, only if
all closer disks to $d_i$ than $d_j$ are also merged with $d_i$.  This
condition, in addition to the second condition of
Definition~\ref{dproper} (disk $d_i$ may be merged with $d_j$ only
if $d_j$ contains the centre of $d_i$ after disks closer to $d_j$
than $d_i$ are merged with $d_j$), ensures the locality of the merges.  By
requiring this ordering for merges, however, we get instances for
which there is no solution, such as the one demonstrated in
Figure~\ref{rot2}.  For such instances, a solution can be obtained
by relaxing this condition.
In this section, we relax the first condition of Definition~\ref{dproper}.

\begin{definition}
In an assignment $\phi$ for a set of disks $D$, let $\delta_i^\phi$
denote the sequence of disks assigned to selected disk $d_i$,
ordered by their distance to $d_i$.  Also, let $\delta_i^\phi(j)$ denote
the $j$-th disk in this sequence.
\end{definition}

\begin{definition}
An assignment $\phi$ is \emph{uproper} (short for unordered proper)
if it meets the following conditions.
\begin{enumerate}
\item
For each pair of possible indices $i$ and $j$, in which
$\phi(d_j) = d_i$, choose $k$ such that $\delta_i^\phi(k) = d_j$.
The distance between $d_i$ and $d_j$ must be at most $r_i + \sum_{x=1}^{k-1} r_{\delta_i^\phi(x)}$.
In other words, after merging all closer disks in $\delta_i^\phi$,
$d_i$ must contain the centre of $d_j$.
\item
Selected disks must be centre-disjoint with respect to their aggregate radii;
i.e.~none of them can contain the centre of any other selected disk.
\end{enumerate}
\end{definition}

\begin{definition}
\label{drelaxed}
Given a set of disks, the goal in the Relaxed Maximum Centre-Disjoint Mergeable Disks
Problem (\textsc{RMCMD}) is to find a uproper assignment of the maximum possible cardinality.
\end{definition}

Theorem~\ref{texists} shows that any set of disks has a uproper
assignment, and therefore, \textsc{RMCMD} always has a solution.

\begin{theorem}
\label{texists}
There exists at least one uproper assignment for any set of disks $D$.
\end{theorem}
\begin{proof}
Let $S$ be the largest subset of $D$ for which there exists a uproper
assignment $\phi_S$.  If $S = D$, we are done; so we assume otherwise.
Let $d_k$ be any disk in $D \setminus S$.  The
centre of $d_k$ is contained in at least one disk $d_i$ of $S$
(considering its aggregate radius); otherwise, $d_k$ may be added to
$\phi$ as a selected disk, which contradicts the choice of $S$.

We obtain a uproper assignment $\phi_{S'}$ for $S' = S \cup \{d_k\}$ as follows.
Initially we set $\phi_{S'}(d_x) = \phi_S(d_x)$ for every $d_x \in S$.
We also set $\phi_{S'}(d_k) = d_i$, as $d_i$ contains the centre of $d_k$.
Merging $d_k$ increases the aggregate radius of $d_i$.
If $\phi_{S'}$ is not uproper, $d_i$ must contain the centre of another
disk $d_j$ such that $\phi_{S'}(d_j) = d_j$ (it must be a selected disk in $\phi_{S}$).
We modify $\phi_{S'}$ so that
$\phi_{S'}(d_x) = d_i$ for every $d_x \in \delta_j^\phi \cup \{d_j\}$ (here,
with some abuse of notation, assume $\delta_j^\phi$ to be a set).
These changes satisfy the second condition of Definition~\ref{drelaxed}:
after merging $d_j$ with $d_i$, $d_i$ contains $d_j$ completely
(because the centre of $d_j$ was inside $d_i$ before the merge and
after it, the aggregate radius of $d_i$ is increased by $r_j$),
and therefore it must contain the centre of $\delta_j^\phi(1)$.
We then merge $\delta_j^\phi(1)$ with $d_i$,
then we can merge $\delta_j^\phi(2)$ with $d_i$, and so on.

\medskip
If, after these changes, $\phi_{S'}$ is not uproper, then $d_i$ contains
the centre of another selected disk $d_y$.  For every such disk,
we similarly merge with $d_i$, $d_y$ and every disk merged with it,
until $d_i$ does not contain the centre of any other disk, and
then, $\phi_{S'}$ becomes proper.  This contradicts the choice of $S$
and implies that we have a uproper assignment for $D$.
\end{proof}

To show that \textsc{RMCMD} is NP-hard, in Theorem~\ref{tproper}
we reduce the \textsc{Partition} problem to \textsc{RMCMD}.
In \textsc{Partition}, we are given a set of positive integers and
have to decide if there is a subset, whose sum is half of the
sum of all numbers in the input list.  \textsc{Partition} is known
to be NP-complete \cite{karp72}.

\begin{theorem}
\label{tproper}
RMCMD is NP-hard.
\end{theorem}
\begin{proof}
We reduce \textsc{Partition} to \textsc{RMCMD}.
Let $A = \{a_1, a_2, \dots , a_n\}$ be an instance of \textsc{Partition} and
$s$ be the sum of the members of $A$.
Also let $e$ be a real number such that $0 < e < 1$; we use $e$ to
create distances smaller than a segment of unit length.
We create an instance of \textsc{RMCMD} as follows (see Figure~\ref{rotg}).
\begin{enumerate}
\item Add disk $d_1$ of radius $2s$ and add $d_2$ with the same radius
at distance $3s$ on the right of $d_1$.
\item Add $d_3$ at distance $5s/2 + e$ above $d_1$ with radius $s$.
Similarly, add $d_4$ at distance $5s/2 + e$ above $d_2$ with the
same radius.
\item
Add one disk for each member of $A$ in the midpoint of the centres of $d_1$ and $d_2$,
such that the radius of the one corresponding to $a_i$ is $a_i$.
\end{enumerate}

\begin{figure}[h]
	\centering
	\includegraphics[width=8cm]{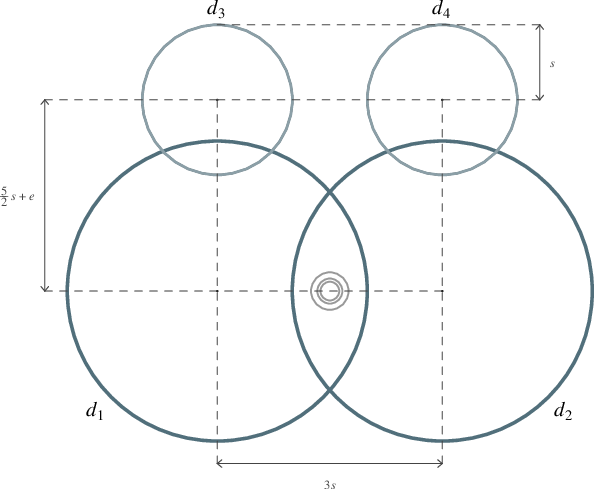}
	\caption{The construction in the proof of Theorem~\ref{tproper}.
	Only if we have a uproper assignment of cardinality four,
	input numbers can be divided into two partitions of equal sum.}
	\label{rotg}
\end{figure}
Let $\phi$ be the solution of this \textsc{RMCMD} instance.
We show that there is a valid solution to the \textsc{Partition} instance
if and only if the cardinality of $\phi$ is four.

Suppose $X$ is a subset of $A$ with sum $s/2$.
We obtain an assignment from $X$ as follows: every disk corresponding
to a member of $X$ is assigned to $d_1$ and others are assigned to
$d_2$.  Since the sum of the members of $X$ is $s/2$, the aggregate
radii of both disks are exactly $5s/2$.  Therefore, the centre of
$d_3$ and $d_4$ are outside these disks.  This yields a uproper
assignment of cardinality 4.

For the reverse direction, suppose the cardinality of $\phi$ is four (note
that it cannot be greater).  If so, all of $d_1$, $d_2$, $d_3$, and
$d_4$ are selected, and therefore, the aggregate radii of $d_1$ and
$d_2$ are lower than $5s/2 + e$.  Given that the sum of the radii
of the disks corresponding to members of $A$ is $s$ (which is an
integer) and $0 < e < 1$, the sum of the set of disks assigned
to $d_1$ and $d_2$
(and therefore the subsets of $A$ corresponding to them) are equal.
\end{proof}

We cannot extend Theorem~\ref{tproper} to the case in which
all disks have the same radius.  The problem is that the radii
of the disks may be arbitrarily large.  This is necessary
in Theorem~\ref{tproper}, because the partition problem is
weakly NP-complete, and using fixed-radius disks (the technique
we used in Corollary~\ref{csame} for MCMD) to construct larger disks may
imply a number of disks that is not polynomial in the input size.

\section{Algorithms for MCMD}
\label{salg}
In this section we present algorithms for solving MCMD.
We present an ILP formulation for general MCMD instances in Section~\ref{silp},
and in Section~\ref{s15} we present a dynamic programming algorithm
for MCMD instances in which disk centres are collinear.

\subsection{ILP formulation of MCMD}
\label{silp}
\begin{theorem}
\label{tilp}
Any instance of MCMD with $n$ disks can be formulated as an integer
linear programme with $O(n^2)$ binary variables.
\end{theorem}
\begin{proof}
For $1 \le i, j \le n$, we introduce a binary
variable $x_{i,j}$.  If $i \ne j$, $x_{i,j}$ indicates whether
disk $d_i$ is merged with disk $d_j$.
If $i = j$, it shows if disk $d_i$ is a selected disk.
The following constraint for $1 \le i \le n$ makes sure
that each disk is either selected or merged with another disk.
\begin{equation}
\sum_{j=1}^{n} x_{i,j} = 1
\end{equation}
An assignment $\phi$ can be obtained from the values of variables $x_{i,j}$
by letting $\phi(d_i) = d_j$ if and only if $x_{i,j} = 1$.

\medskip
Based on the conditions enumerated in Definition~\ref{dproper},
we add additional constraints to make sure that the obtained
assignment $\phi$ is proper.
Let $\delta_i$ be as defined in Definition~\ref{ddelta}.
\begin{enumerate}
\item
Based on the first condition of Definition~\ref{ddelta}, a disk can be
merged with $d_i$, only if its closer disks are merged as well.
In other words, $x_{\delta_i(j+1)i}$ can be one, only if $x_{\delta_i(j)i}$ is
also one for $1 \le j \le n - 2$, as expressed in the following constraint
for $1 \le i \le n$ and $1 \le j \le n - 2$.
\begin{equation}
x_{\delta_{i}(j) i} \ge x_{\delta{i}(j+1) i}
\end{equation}
\item
Based on the second condition of Definition~\ref{ddelta}, $d_j$ can be
merged with $d_i$, only if the distance of $d_i$ and $d_j$ is less
than $r_i(j - 1)$ (Definition~\ref{ddelta}) for $1 \le j \le n - 1$;
that is, after merging with $d_i$ all disks closer than $d_j$ to $d_i$,
the resulting disk must contain the centre of $d_j$.
We disallow merges that fail this condition:
for $1 \le i \le n$ and $1 \le j \le n - k$, where $r_i(j - 1) < |p_ip_j|$,
we add the following constraint (for simplicity, let $r_i(0) = r_i$).
\begin{equation}
\label{elim}
x_{ji} = 0
\end{equation}
\item
Based on the third condition of Definition~\ref{ddelta}, selected disks
must be centre-disjoint.  We add the following constraint for $1 \le i, j \le n$.
\begin{equation}
\label{edisjoint}
\sum_{k=1}^{n} r_k \cdot x_{ki}
\le |p_ip_j| + \infty \cdot (1 - x_{jj})
\end{equation}
Obviously, the left side computes the aggregate radius of $d_i$;
note that if $d_i$ is not selected, the left side equals zero and
the inequality is trivially satisfied.  There are two cases based on
whether $d_j$ is selected or not.  If $d_j$ is a selected disk ($x_{jj} = 1$)
the right side of the inequality simplifies to $|p_ip_j|$, making sure
that $d_i$ does not contain $d_j$.  If, on the other hand, $d_j$ is merged
with another disk (maybe even with $d_i$), the right side of the inequality
simplifies to $+\infty$ and the constraint is satisfied.
\end{enumerate}

Finally, as the goal of MCMD (Definition~\ref{dmcmd}) is to maximize the
number of selected disks, the objective of the programme is simply
to maximise $\sum_{i=1}^{n} x_{ii}$.
\end{proof}

The number of variables used in the integer programme of Theorem~\ref{tilp}
can be reduced based on the following observation.
The value of some variables of an MCMD instance is always zero by
constraints of type~\ref{elim}.  These variables can be removed.
The implementation of the ILP of Theorem~\ref{tilp} with
this optimization is publicly available\footnote{\url{https://github.com/nit-ce/mcmd.git}};
it has been used to obtain Figure~\ref{rotx}.

\subsection{Collinear disk centres}
\label{s15}
In this section we present a polynomial-time algorithm for solving
MCMD for a set of disk with collinear centres.
Note that even if disk centres are collinear, there may exist
no proper assignments, as demonstrated in Figure~\ref{rot2}.
In the rest of this section, let $\lambda = \left<d_1, d_2, \dots , d_n\right>$
be a sequence of input disks $D$, ordered by the $x$-coordinate of their centres.
We assume that the centres of the members of $D$ are collinear
and are on the $y$-axis.  We need Definitions~\ref{dext} and \ref{dsel}
to present the algorithm in Theorem~\ref{t15}.

\begin{definition}
\label{dext}
Let $\phi$ be an assignment of $\{d_1, d_2, \dots , d_n\}$
and let $\phi'$ be an assignment of $\{d_1$, $d_2, \dots , d_x\}$,
such that $x \le n$.
$\phi$ is \emph{an extension} of $\phi'$, if for every disk $d_i$
in $\{d_1, d_2, \dots , d_x\}$, we have $\phi(d_i) = \phi'(d_i)$.
In other words, every selected disk in $\phi'$ is also a selected
disk in $\phi$, and every merged disk in $\phi'$ is also merged
with the same disk in $\phi$.
Equivalently, when $\phi$ \emph{is limited to} $\{d_1, d_2, \dots , d_x\}$,
$\phi'$ is obtained.
\end{definition}

\begin{definition}
\label{dsel}
$M(x, y, z)$ denotes the maximum cardinality of a
proper assignment of $X = \{d_1, d_2$, $\dots,  d_x\}$, such that the
following conditions are met (we have $y \le x \le z \le n$).
\begin{enumerate}
\itemsep=0.8pt
\item $d_y$ is its right-most selected disk.
\item $d_{y+1}, d_{y+2}, \dots , d_{x}$ are all merged with $d_y$.
\item $d_z$ is the right-most disk in $D$, where $z \ge x$,
whose centre is contained in $d_y$ considering its aggregate
radius.
\end{enumerate}
\end{definition}
Note that by the third condition of Definition~\ref{dsel},
the centres of $d_{x+1},\allowbreak d_{x+2},\allowbreak \dots ,\allowbreak d_z$ are inside $d_y$
with respect to $D$,
but they are not merged with it, because they are outside $X$
and not present in the assignment which is limited to set $X$.
Also, note that actually the second condition of Definition~\ref{dsel}
is implied by its first condition: since $d_y$ is the right-most
selected disk, all of the disks that appear on the right
of $d_y$ in $X$ are surely merged.  On the other hand, none of
these disks can be merged with a selected disk $d_w$ on the
left of $d_y$, because, in that case $d_w$ would contain
the centre of $d_y$ and the assignment cannot be proper.

\begin{theorem}
\label{t15}
A proper assignment of the maximum cardinality for a set of
$n$ disks $D$, in which the centres of all disks are collinear,
can be computed in polynomial time.
\end{theorem}

\begin{proof}
Let $M$ be defined as in Definition~\ref{dsel}.
Obviously, $\max_{i=1}^n M(n, i, n)$ is the cardinality of the solution
to this \textsc{MCMD} instance.

\medskip
The function $M$ accepts $O(n^3)$ different input values.
We can compute and store the values returned by $M$ in a three
dimensional table, which we reference also as $M$.
Algorithm~\ref{amain} uses dynamic programming to fill $M$,
and computes its entries based on the values of previously computed entries.
Before explaining the details of the algorithm, we give a high-level
overview as follows.  The order of the disks referenced here, in the algorithm,
and its succeeding discussion is demonstrated in Figure~\ref{f15}.

The main idea behind the algorithm is that in every proper
assignment of $\{d_1, d_2, \dots , d_x\}$, there is at least one
selected disk; take the right-most selected disk $d_i$.
By the first condition of Definition~\ref{dproper}, for some $j$
where $0 \le j \le n - 1$, the first $j$ entries of $\delta_i$
are merged with $d_i$.  Thus, the algorithm considers different
values of $j$ for each disk $d_i$, to compute the maximum proper
assignment of $\{d_1, d_2, \dots , d_x\}$ for any $x$
in which $d_i$ is the right-most selected disk and $j$ disks
are merged with $d_i$; it updates the entries of $M$ accordingly.
The main challenge is to
decide what to do with the disks that appear on the left of
$d_i$ and use the previously filled entries of $M$ for them.
To do so, the algorithm enumerates possible choices for the right-most
selected disk $d_t$ that appear on the left of $d_i$, and the
number of disks $k$ merged with it.

After presenting Algorithm~\ref{amain}, we shall explain the steps
of this algorithm in more detail.

Steps~\ref{i1} and \ref{i2} of the algorithm initialize $M$ and $\delta_i$.
In Step~\ref{i3}, we consider different cases in which $d_i$,
for $1 \le i \le n$ in order, is selected and update the
value of different entries of $M$.
For every possible value of $j$ from $0$ to $n - 1$,
suppose $j$ disks are merged with $d_i$.
These disks are the first $j$ disks of $\delta_i$
by the first condition of Definition~\ref{dproper}.

\clearpage

\begin{figure}[!b]
\vspace*{-5mm}
	\centering
	\includegraphics[width=11cm]{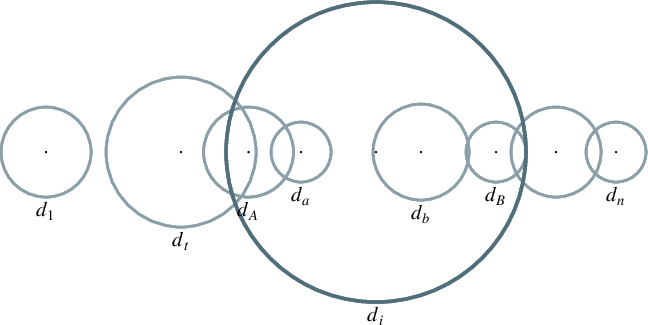}
	\caption{Demonstrating the symbols used in Theorem~\ref{t15}.
	$d_i$ is the right-most selected disk, $\{d_a , \dots d_b \} \setminus \{d_i\}$
	are merged with $d_i$, and $d_i$ contains the centre of $\{d_A, \dots , d_B\}$.
	$d_t$ is the right-most selected disk on the left of $d_i$.
	}
	\label{f15}
\end{figure}

\medskip
\begin{algorithm}[!ht]\small{
\caption{Find a solution to MCMD for a set of collinear disks}
\label{amain}
\begin{enumerate}
\item\label{i1} Compute the sequences $\delta_i$ for $1 \le i \le n$ (Definition~\ref{ddelta}).
\item\label{i2} Initialize every entry of $M$ to $0$.
\item\label{i3} For each $i$ in $1 \le i \le n$, for each $j$ in $0 \le j \le n - 1$ repeat:
\begin{enumerate}
\item\label{i3a}
If the first $j$ disks of $\delta_i$ can not be merged with $d_i$ according
to the second condition of Definition~\ref{dproper},
continue with the next iteration of this loop for the next value
of $i$ and $j$.
\item\label{i3b} Compute $A$, $a$, $b$, and $B$: $a$ and $b$ are the
left-most and right-most disks in $\lambda$ that are merged with $d_i$,
respectively.  Also, $A$ and $B$ are the left-most and right-most
disks of $D$ whose centres are contained in $d_i$, considering its
aggregate radius (note that we have $A \le a \le b \le B$).
\item\label{i3c} If $a = 1$ and $M(b, i, B) == 0$, assign 1 to $M(b, i, B)$.
\item\label{i3d} If $a > 1$,
for $t$ in $1 \le t \le A - 1$,
for $k$ in $0 \le k \le n - 2$ repeat:
\begin{enumerate}
\item\label{i3d1} If the first $k$ disks of $\delta_t$ cannot be merged
with $d_t$ (according to the second condition of Definition~\ref{dproper}),
continue to the next iteration of this loop.
\item\label{i3d2} Compute indices $f$ and $g$: $d_f$ is the right-most disk that
is merged with $d_t$, and $d_g$ is the right-most disk
of $D$ whose centre is contained in $d_t$, considering its
aggregate radius.
\item\label{i3d3} If $f \ge a$, $f \neq a-1$, or $g \ge i$ (disks merged with $d_i$
are being merged with $d_t$ or $d_t$ contains the centre of $d_i$,
both of which are not allowed in proper assignments),
continue to the next iteration of this loop.
\item\label{i3d4} Replace the value of $M(b, i, B)$ with the maximum
of its value and $M(a - 1, t, g) + 1$.
\end{enumerate}
\end{enumerate}
\item Compute and return $\max_{i=1}^n M(n, i, n)$.
\end{enumerate} }\normalsize
\end{algorithm}

\eject

\medskip
Let $S$ denote the set of such disks.
If this is not possible (the centre of one of these disks is not
contained in $d_i$, after merging its previous disks), we skip this value of $j$,
because it fails the second condition of Definition~\ref{dproper} (Step~\ref{i3a}).
Note that if there exists no proper assignment in which $j$ disks
are merged with $d_i$,
there cannot exists a proper assignment in which more than $j$
disks are merged with $d_i$ either, and we can safely skip
the remaining values of $j$ and continue the loop of Step~\ref{i3} by
incrementing the value of $i$.

\medskip
Let $a$, $b$, $A$, and $B$ be defined as Step~\ref{i3b}.
If $a = 1$, selecting $d_i$ and merging with it every disk in
$\{d_1, d_2, \dots , d_b\} \setminus \{d_i\}$ is a proper assignment
of the first $b$ disks of $\lambda$ with cardinality one.  Therefore,
we update the value of $M(b, i, B)$ to be at least one in Step~\ref{i3c}.

\medskip
If $a > 1$, let $\phi$ be any assignment of $\{ d_1, \dots , d_b \}$,
in which i) $d_i$ is selected, ii) the members of $S$ are merged with $d_i$,
and iii) the members of $\{d_{A}, \dots , d_{a - 1} \} \cup \{d_{b+1}, \dots , d_{B}\}$
are contained in $d_i$ after merging the members of $S$ with $d_i$.
By the definition of $M$, the value of $M(b, i, B)$ cannot be smaller
than the cardinality of $\phi$.
When $\phi$ is limited to $L = \{d_1, \dots , d_{a-1}\}$, it specifies a
proper assignment of $L$.  We denote this assignment with $\phi_L$.
We compute the value of $M(b, i, B)$ by considering all possible
assignments for $\phi_L$ and extending them to obtain $\phi$ by selecting $d_i$.

\medskip
Let $d_t$ be the right-most selected disk of $\phi_L$.
The following conditions hold.
\begin{enumerate}
\item We have $t < A$, because $\{d_{A}, \dots , d_{a - 1}\}$ are contained in $d_i$
in $\phi$, and $d_t$ cannot be a selected disk if $t \ge A$.
Therefore, disks $\{d_{t+1}, \dots , d_{a - 1}\}$ are merged with $d_t$ in $\phi_L$.
\item Suppose $k$ disks are merged with $d_t$ in $\phi_L$.
Let $d_f$ be the right-most disk of $D$ contained in $d_t$
after merging disks in $\phi_L$.  We have $f < i$; otherwise,
$d_f$ would contain the centre of $d_i$, and $d_i$ cannot be selected in $\phi$.
Also let $d_g$ be the right-most vertex of $D$ contained in $d_f$.
We have $g < i$; otherwise, $d_f$ would contain the centre of $d_i$ and
$\phi$ cannot be an extension of $\phi_L$.
\end{enumerate}
By trying possible values of $t$ and $k$ that meet these conditions
(Step~\ref{i3d}), we find the maximum cardinality of $\phi_L$, which
has been computed in the previous steps of this algorithm as $M(a - 1, t, g)$.
Since $\phi$ is an extension of $\phi_L$ by adding exactly one
selected disk $d_i$, the maximum cardinality of $\phi$ therefore
is at least $1 + M(a - 1, t, g)$.
Thus, we have
\[
M(b, i, B) \ge 1 +
\max_{{t: 1 \le t \le A - 1} \atop
{{k: 0 \le k \le n - 2} \atop
{\textit{if}~3.d.i~\textit{and}~3.d.iii~\textit{hold}}}} M(a - 1, t, g)
\]
Step~\ref{i3d4} updates $M(b, i, B)$ to be at least this value.
\end{proof}

\begin{theorem}
\label{t16}
The time complexity of computing $M$ for an instance of MCMD
with a set of $n$ disks, as described in Theorem~\ref{t15}, is $O(n^5)$.
\end{theorem}
\begin{proof}
We analyse Algorithm~\ref{amain}.
Constructing $\delta_i$ (Step~\ref{i1}) can be done in $O(n^2 \log n)$
and initializing $M$ (Step~\ref{i2}) can be done in $O(n^3)$.
For each pair of values for $i$ and $j$, Steps~\ref{i3a}-\ref{i3c}
can be performed in $O(n)$.
In Step~\ref{i3d}, $O(n^2)$ possible cases for $t$ and $k$ are
considered, and for each of these cases, the Steps~\ref{i3d1},
\ref{i3d2}, \ref{i3d3}, and \ref{i3d4} can be performed in $O(n)$.
Since the loop of Step~\ref{i3} is repeated $O(n^2)$ times,
the time complexity of the whole algorithm is $O(n^5)$.
\end{proof}

Note that Algorithm~\ref{amain} cannot be used for
RMCMD (Defintion~\ref{drelaxed}) as it relies on the
first condition of Definition~\ref{dproper}: $d_j$ may be merged
with $d_i$, only if disks closer to $d_i$ than
$d_j$ are also merged with it.  This does not hold for RMCMD.

\section{Concluding remarks}
\label{scon}
We introduced a variant of geometric independent set for a
set of disks, such that the disks that do not appear in the
output must be merged with a nearby disk that does
(the problem was stated formally in Section~\ref{spre}).
We proved that this problem is NP-hard (Theorem~\ref{tnp}).
Also by relaxing one of the conditions of the problem, we introduced a
less restricted variant, which was proved NP-hard as well (Theorem~\ref{tproper}).
We presented an ILP for the general case, and
a polynomial-time algorithm for the case in which
disk centres are collinear.

The centre-disjointness property of the disks in the definition
of MCMD and RMCMD implies that we are implicitly assuming
Euclidean distance between the centre of the disks;
disk $d_i$ containing the centre of $d_j$ implies that
the Euclidean distance between $p_i$ and $p_j$ is at
most $r_i$.  If the problem is defined using other distance
measures, we would have different shapes; for instance squares
for $L_\infty$ measure.  We can even define the problem using
any (maybe irregular) shape, or in higher dimensions using,
for instance, spheres.
Interestingly, the hardness results of Section~\ref{tnp} can be
adapted to centre-disjoint squares, after adjusting the gadgets
(using squares instead of disks and updating the placement of
m-disks slightly).  Other shapes (or distance measures) may
be studied for hardness results or better algorithms.

\medskip
Several interesting problems call for further investigation, such as:
i) general approximation algorithms,
ii) studying the case in which the number of disks
that can be merged with a selected disk is bounded by some constant,
iii) solving RMCMD for disks with collinear centres, and
iv) solving MCMD when disks are pierced by a line (the line may
not pass through disk centres).

\end{document}